\documentclass[11pt]{article}

\usepackage[margin=1.1in]{geometry}
\usepackage{amsmath,amssymb,amsthm}
\usepackage{booktabs}
\usepackage{array}
\usepackage{microtype}
\usepackage{listings}
\usepackage{xcolor}
\usepackage[colorlinks=true,linkcolor=blue!60!black,citecolor=blue!60!black,urlcolor=blue!60!black]{hyperref}
\usepackage[round]{natbib}
\usepackage{needspace}
\usepackage{tabularx}
\usepackage{tikz}
\usetikzlibrary{arrows.meta}

\newtheorem{theorem}{Theorem}[section]
\newtheorem{proposition}[theorem]{Proposition}
\newtheorem{corollary}[theorem]{Corollary}

\theoremstyle{definition}
\newtheorem{definition}[theorem]{Definition}

\theoremstyle{remark}
\newtheorem{remark}[theorem]{Remark}

\newcommand{\FOUR}{\mathsf{FOUR}}
\newcommand{\vY}{Y}
\newcommand{\vN}{N}
\newcommand{\vYN}{Y\!N}
\newcommand{\vU}{U}
\newcommand{\lek}{\le_k}
\newcommand{\leqt}{\le_t}
\newcommand{\Sopen}{S_{\mathrm{open}}}
\newcommand{\Sred}{S_{\mathrm{reduced}}}
\newcommand{\Squasi}{S_{\mathrm{quasi}}}
\newcommand{\dcons}{\delta_{\mathrm{cons}}}
\newcommand{\dprec}{\delta_{\mathrm{prec}}}
\newcommand{\dopt}{\delta_{\mathrm{opt}}}
\newcommand{\dcont}{\delta_{\mathrm{cont}}}
\newcommand{\Ent}{\mathcal{E}}
\newcommand{\warr}{\mathbin{\Rightarrow_{\!w}}}

\lstdefinelanguage{lean}{
  keywords={theorem,def,inductive,structure,namespace,end,variable,instance,
            match,with,by,fun,cases,induction,exact,intro,obtain,where,deriving},
  sensitive=true,
  comment=[l]{--},
  morecomment=[s]{/-}{-/},
}
\title{\textbf{A Four-Valued Graph Model for Conflict Resolution:\\
Core Framework and a Machine-Checked Formalization in Lean~4}}
\author{Yukiko Kato\\[2pt]
\normalsize Institute of Science Tokyo, Japan\\
\normalsize \texttt{ykato@lynxjp.com}}
\date{}

\begin{document}
\date{}
\maketitle

\begin{abstract}
This note consolidates the core of the Quasi-Closed World Graph Model for
Conflict Resolution (QCW-GMCR), a framework that extends the standard Graph
Model for Conflict Resolution (GMCR) with Belnap's four-valued logic in order
to represent epistemic ambiguity at the option level, and pairs it with a
machine-checked formalization in the Lean~4 proof assistant. The framework
comprises four components: (1) option-level epistemic assignments in
$\FOUR=\{\vY,\vN,\vYN,\vU\}$ with compositional propagation to state-level
feasibility; (2) graded reachability (definite/credible/possible) derived from an FDE-inspired
transition-warrant semantics, with definite reachability related explicitly to
FDE consequence in the Boolean fragment; (3) formally constrained reduction operators
mapping four-valued assignments to binary decisions, with four canonical
operators characterized by distinct risk attitudes; and (4)
catastrophe-avoiding equilibrium concepts with a structural safety guarantee,
the quasi-closed world invariant. A four-valued hypergame extension embeds the
framework in the hypergame paradigm and strictly generalizes it along the
epistemic axis. We state the main definitions and theorems in
a compact form and report which results have been verified in Lean~4 with
mathlib, including the classical GMCR stability hierarchy
($\mathrm{Nash}\subseteq\mathrm{SMR}\subseteq\mathrm{GMR}$ and
$\mathrm{Nash}\subseteq\mathrm{SEQ}\subseteq\mathrm{GMR}$), both for the
legal move sequences of the standard definitions and for unrestricted
sequences, the algebraic
properties of four-valued conjunction, the compositional propagation
properties, the classical-extension and truth-monotonicity properties of the
canonical reduction operators, and the
graded reachability hierarchy. The Lean development is available in the
accompanying repository. This preprint is intended as a stable, citable record
of the framework's skeleton and its formal verification status rather than as
a full journal exposition.
\end{abstract}


\section{Introduction}\label{sec:intro}

The Graph Model for Conflict Resolution (GMCR) analyzes strategic conflicts
among decision-makers (DMs) using ordinal preferences over discrete states and
graph-based reachability \citep{kilgour1987,fang1993}. Standard GMCR inherits
a closed-world, binary epistemology: every option is either selected or not,
every state is either feasible or not, and every transition is either possible
or not. This is analytically convenient but fails to represent situations in
which the analyst's information is genuinely incomplete or contradictory---for
example, when it is unknown whether an opponent has committed to an option, or
when intelligence about a capability is contradictory. The distinction between
``$B$ will not act'' and ``we do not know whether $B$ will act'' matters for
rational choice, yet binary models collapse it.

This gap matters most for grave, irreversible decisions that resist reliable
quantification because controlled experiments, repeated observations, or defensible
probability distributions are unavailable. This motivation is related to the broader
literature on decision making under deep or severe uncertainty, where robust action
must sometimes be chosen without a well-founded probabilistic model
\citep{benhaim2019}. QCW-GMCR takes a different route: rather than assigning a
numerical degree to such uncertainty, it represents epistemic status directly as a
grounded \emph{yes} or \emph{no} ($\vY$, $\vN$), conflicting grounds
warranting both ($\vYN$), or evidence warranting neither ($\vU$)---exactly
Belnap's four values. This distinction motivates the framework's use of four-valued
logic and its catastrophe-avoiding design.

The Quasi-Closed World GMCR (QCW-GMCR) addresses this by grounding option
status in Belnap's four-valued logic \citep{belnap1977a,belnap1977b},
building on the four-valued state definition introduced in \citet{kato2022}.
Relative to that work, which concerned the \emph{state}-level definition, the
present framework contributes the option-level compositional grounding, the
entailment-based graded reachability, the formally constrained reduction operators, and
the catastrophe-avoiding equilibria with their safety invariant; the
machine-checked formalization is new to this note.

Various GMCR extensions treat uncertainty through fuzzy preferences
\citep{bashar2012}, grey and unknown preferences \citep{yang2022},
probabilistic preferences \citep{rego2015}, or preference uncertainty in
coalitions \citep{li2014}, and the hypergame tradition models misperception
through divergent subjective games
\citep{bennett1977,wang1988,wanghf1989}, with a developed GMCR integration
ranging from first-level and general hypergames in graph form
\citep{aljefri2018,aljefri2020} to recent combinations of hypergames with
uncertain preference misperception and algebraic stability representations
\citep{zhao2024}. A further line models \emph{interactive unawareness}: a DM
may be unaware of some options available to itself or others, with iterated
levels of awareness about what other DMs are aware of \citep{rego2020}. These
lines locate uncertainty in \emph{what is valued}, in \emph{which game each
DM perceives}, or in \emph{which options a DM is aware of}. The four-valued
status $\psi_{i,o}$ concerns something prior to all three: whether the
holding of an option or state is epistemically established at all, independent
of any preference or likelihood information---so option-status ambiguity,
including the contradictory value $\vYN$, is not an object that uncertain
preferences or numerical transition probabilities are about. The value $\vU$
(neither $\vY$ nor $\vN$) records the absence of grounds for either verdict
for an option represented in the DM's subjective model. This is distinct from,
but complementary to, the structural unawareness studied by \citet{rego2020},
where an option may be absent from the DM's subjective model altogether: thus
QCW-GMCR grades the epistemic standing of options already represented, while
interactive-unawareness models determine which options enter the subjective
game in the first place. The relation to hypergames is closer and complementary; we make the
correspondence precise in Section~\ref{sec:hyper}
(Remark~\ref{rem:hyperrel}).

This note has two purposes. First, it presents the skeleton of QCW-GMCR---the
definitions, propositions, and safety guarantees that constitute the
framework---in a compact, self-contained form, stripped of extended case
studies and disciplinary positioning. Second, it documents a machine-checked
formalization of the framework's core in the Lean~4 proof assistant
\citep{demoura2021} with the mathlib library \citep{mathlib2020}. Formal
verification is a natural companion to a framework whose stated purpose is
safety under epistemic ambiguity: the guarantees offered to an analyst should
themselves be mechanically certain. Section~\ref{sec:lean} reports precisely
which results are verified and which remain informal, together with the
correspondence between paper statements and Lean identifiers. The Lean
development is available at:
\begin{center}
  \url{https://github.com/ykato0623/qcw-gmcr-lean} \quad 
\end{center}

\section{Preliminaries}\label{sec:prelim}

\subsection{Standard GMCR}

\begin{definition}[GMCR structure]\label{def:gmcr}
A \emph{graph model} is a tuple
$\mathcal{G} = (N, S, \{A_i\}_{i\in N}, \{\succ_i\}_{i\in N})$, where $N$ is a
finite set of DMs; $S$ is a finite set of states; $A_i \subseteq S\times S$ is
DM $i$'s set of unilateral moves, written $s \to_i s'$; and $\succ_i$ is
DM $i$'s \emph{strict} preference relation on $S$. As in standard GMCR, each
DM's preference is \emph{complete} (any two states are comparable, so it is
not merely a partial order) but need \emph{not} be transitive, and it admits
indifference: $s \sim_i t$ arises exactly when neither is strictly preferred,
$\lnot(s \succ_i t)$ and $\lnot(t \succ_i s)$. Only \emph{strict} preference
enters the notion of improvement. A pair $(s,s)\in A_i$ may be used to record that DM $i$ can remain
at $s$; remaining is not a move, and reachable lists are defined over proper
moves only. DM $i$'s \emph{reachable list} and \emph{unilateral improvement
set} at $s$ are
$R_i(s) = \{ s' \in S : (s,s') \in A_i,\ s' \neq s \}$ and
$R^+_i(s) = \{ s' \in R_i(s) : s' \succ_i s \}$.
For the stability definitions below, we write $s \succsim_i t$ as shorthand
for $\lnot(t \succ_i s)$---``the move to $t$ is not a strict improvement over
$s$ for $i$'', equivalently $s \succ_i t$ or $s \sim_i t$. We note in passing
that the stability definitions and the machine-checked hierarchy theorems
below use only this shorthand (a sanction to $u$ is required to satisfy
$\lnot(u \succ_i s)$), so they impose no axioms on $\succ_i$ at all---in
particular neither completeness nor transitivity---and the standard
complete (possibly intransitive) preferences of GMCR are a special case.
\end{definition}

\begin{definition}[Coalition reachability {\citep{fang1993}}]\label{def:reach}
Let $H\subseteq N$ be nonempty and $s\in S$. The set $R_H(s)\subseteq S$ and,
for each $t\in S$, the set $\Omega_H(s,t)\subseteq H$ of \emph{last movers} are
defined simultaneously, by induction, as the least sets satisfying:
(i) if $j\in H$ and $t\in R_j(s)$, then $t\in R_H(s)$ and
$j\in\Omega_H(s,t)$; and (ii) if $t'\in R_H(s)$, $j\in H$, $t\in R_j(t')$, and there exists
$k\in\Omega_H(s,t')$ with $k\neq j$, then $t\in R_H(s)$ and
$j\in\Omega_H(s,t)$.
The sets $R^+_H(s)$ and $\Omega^+_H(s,t)$ are defined in the same way with
$R_j$ replaced by $R^+_j$.
\end{definition}

Condition (ii) is the legality constraint of \citet{fang1993}: the same DM
may move more than once, but not twice in succession. Equivalently, $t\in
R_H(s)$ iff some nonempty sequence of moves by members of $H$ leads from $s$
to $t$ with no DM moving twice in succession, and $\Omega_H(s,t)$ is the set
of DMs that make the final move of at least one such sequence. Thus
$\Omega_H(s,t)$ is not an attribute of the state $t$: it depends on $s$ and
on the legal sequences from $s$ to $t$, and a state reached by several such
sequences may have several last movers. Condition (ii) lets $j$ move from
$t'$ as long as at least one legal sequence reaches $t'$ through a final
move by some other DM. For $H=\{j\}$, $R_H(s)$ reduces to $R_j(s)$. Because reachable lists contain
proper moves only, a DM's recorded ability to remain at a state cannot be used
to interleave and so circumvent the constraint.

\begin{definition}[Stability concepts]\label{def:stability}
Let $s\in S$ and $i\in N$, and write $N-i$ for $N\setminus\{i\}$. State $s$ is
\emph{Nash stable} for $i$ if $R^+_i(s)=\emptyset$;
\emph{GMR stable} for $i$ if for every $s_1\in R^+_i(s)$ there is
$s_2\in R_{N-i}(s_1)$ with $s\succsim_i s_2$;
\emph{SMR stable} for $i$ if for every $s_1\in R^+_i(s)$ there is
$s_2\in R_{N-i}(s_1)$ with $s\succsim_i s_2$ and $s\succsim_i s_3$ for every
$s_3\in R_i(s_2)$; and
\emph{SEQ stable} for $i$ if for every $s_1\in R^+_i(s)$ there is
$s_2\in R^+_{N-i}(s_1)$ with $s\succsim_i s_2$.
A state stable for all DMs under a concept is an \emph{equilibrium} of that
concept.
\end{definition}

The classical inclusions
$\mathrm{Nash}\subseteq\mathrm{SMR}\subseteq\mathrm{GMR}$ and
$\mathrm{Nash}\subseteq\mathrm{SEQ}\subseteq\mathrm{GMR}$ \citep{fang1993}
are among the results verified in Lean (Section~\ref{sec:lean}), both for the
legal reachability of Definition~\ref{def:reach} and for the unrestricted
reachability in which the same DM may move twice in succession. The two
notions differ when a DM's moves are not transitively closed; stability with
legal reachability implies stability with unrestricted reachability, but not
conversely (Section~\ref{sec:lean}). Notably, the
verification uses only the shorthand of Definition~\ref{def:gmcr} and imposes
\emph{no} axioms on $\succ_i$, so the theorems hold without assuming
completeness or transitivity, and in particular under the standard complete,
possibly intransitive preferences of GMCR.

\subsection{Belnap's four-valued logic}

Belnap--Dunn four-valued logic, also known as first-degree entailment (FDE),
has a substantial literature beyond the original presentations; for a modern
overview and historical clarification see \citet{omori2017}.

\begin{definition}[Truth space and bilattice]\label{def:four}
Belnap's truth space is $\FOUR = \{\vY,\vN,\vYN,\vU\}$: true ($\vY$),
false ($\vN$), both/contradictory ($\vYN$), and neither/unknown ($\vU$).
$\FOUR$ carries two partial orders: the \emph{truth order} $\leqt$ with
$\vN <_t \vU <_t \vY$ and $\vN <_t \vYN <_t \vY$ ($\vU,\vYN$ incomparable),
and the \emph{knowledge order} $\lek$ with $\vU <_k \vN,\vY <_k \vYN$.
\end{definition}

\begin{definition}[Connectives]\label{def:connectives}
Negation swaps evidence: $\lnot \vY = \vN$, $\lnot\vN=\vY$,
$\lnot\vYN=\vYN$, $\lnot\vU=\vU$. Conjunction $\wedge$ is the meet and
disjunction $\vee$ the join in the truth order. Table~\ref{tab:conj} gives
the conjunction table used throughout.
\end{definition}

\begin{table}[ht]
\centering
\caption{Four-valued conjunction.}
\label{tab:conj}
\begin{tabular}{c|cccc}
\toprule
$\wedge$ & $\vY$ & $\vN$ & $\vYN$ & $\vU$\\
\midrule
$\vY$  & $\vY$  & $\vN$ & $\vYN$ & $\vU$\\
$\vN$  & $\vN$  & $\vN$ & $\vN$  & $\vN$\\
$\vYN$ & $\vYN$ & $\vN$ & $\vYN$ & $\vN$\\
$\vU$  & $\vU$  & $\vN$ & $\vN$  & $\vU$\\
\bottomrule
\end{tabular}
\end{table}

\begin{remark}[Propagation-relevant algebra]\label{rem:algebra}
$\vN$ is absorbing and $\vY$ is the identity for $\wedge$; conjunction is
commutative, associative, and idempotent; and $\vYN \wedge \vU = \vN$
(incompatible ambiguity types resolve pessimistically). These algebraic facts,
which drive Proposition~\ref{prop:propagation}, are verified in Lean.
\end{remark}

\section{Option-Level Epistemic Foundations}\label{sec:option}

\begin{definition}[Option-level assignment; four-valued state]\label{def:psi}
Let $O_i$ be DM $i$'s option set and $O=\bigcup_{i\in N} O_i$. An
\emph{option-level epistemic assignment} is $\psi_{i,o} \in \FOUR$ for each
$i\in N$, $o\in O_i$: definitely selected/executable ($\vY$), definitely not
($\vN$), contradictory evidence ($\vYN$), or unknown ($\vU$). A
\emph{four-valued state} is a complete assignment
$s \in \FOUR^{|O|}$; the open-world state space is
$\Sopen^{\psi} = \FOUR^{|O|}$, of cardinality $4^{|O|}$ versus $2^{|O|}$ in
standard GMCR.
\end{definition}

\begin{definition}[Compositional feasibility]\label{def:comp}
A \emph{feasibility composition operator} is a map
$\Gamma : \FOUR^{|O|} \to \FOUR$ from option-level assignments to state-level
feasibility $\phi_s$. The default is \emph{conjunctive composition}
$\Gamma_\wedge(s) = \bigwedge_{i,o} \psi_{i,o}(s)$: a state is feasible only
to the degree that all of its constituent options are.
\end{definition}

\begin{proposition}[Compositional propagation]\label{prop:propagation}
Under $\Gamma_\wedge$:
(i) if $\psi_{i,o}=\vN$ for any option, then $\phi_s=\vN$;
(ii) if all $\psi_{i,o}\in\{\vY,\vYN\}$ with at least one $\vYN$, then
$\phi_s=\vYN$;
(iii) if all $\psi_{i,o}\in\{\vY,\vU\}$ with at least one $\vU$, then
$\phi_s=\vU$;
(iv) if all $\psi_{i,o}=\vY$, then $\phi_s=\vY$; and
(v) if some options are $\vYN$ and others $\vU$, then $\phi_s=\vN$.
\end{proposition}

\begin{proof}
Immediate from Table~\ref{tab:conj}; case (v) reflects
$\vYN\wedge\vU=\vN$. Cases (i) and (iv) are verified in Lean over arbitrary
finite option sets (Section~\ref{sec:lean}).
\end{proof}

State-level ambiguity thus \emph{inherits} transparently from option-level
uncertainty, so the analyst can trace which options drive ambiguity. Option
dependence constraints (mutual exclusivity, implication) extend to $\FOUR$ by
evaluating the constraint formula four-valuedly and conjoining it with
$\Gamma(s)$; we omit the routine details here.

\section{Entailment-Based Reachability}\label{sec:reach}

Standard reachability is all-or-nothing. When options carry four-valued
status, transitions differ in \emph{epistemic warrant}. The conceptual core
of the integration proposed here is that this warrant is not an annotation
added to reachability from outside, but is tied to the information semantics of
Belnap--Dunn FDE \citep{anderson1962,omori2017}. Standard FDE is an
implication-free first-degree logic; accordingly, the operation introduced below
is not presented as a primitive FDE conditional. It is an FDE-inspired
transition-warrant operation whose connection to FDE consequence is stated
explicitly in Proposition~\ref{thm:fdechar}.

\begin{definition}[Transition-warrant connective; move warrant]\label{def:entail}
The \emph{transition-warrant connective} $\warr$ on $\FOUR$ is given by
Table~\ref{tab:entail}. It is a transition-specific four-valued operation inspired
by FDE semantics, rather than a primitive conditional of FDE itself. For a
unilateral move $s\to_i s'$, the \emph{move warrant} is
$\Ent(s\to_i s') = \bigwedge_{o\in O_i^{\Delta}}
\bigl(\psi_{i,o}(s) \warr \psi_{i,o}(s')\bigr)$,
where $O_i^{\Delta}$ is the set of options whose value changes.
\end{definition}

\begin{table}[ht]
\centering
\caption{FDE-inspired transition-warrant connective on $\FOUR$ (row: pre-value; column: post-value).}
\label{tab:entail}
\begin{tabular}{c|cccc}
\toprule
$\warr$ & $\vY$ & $\vN$ & $\vYN$ & $\vU$\\
\midrule
$\vY$  & $\vY$  & $\vN$ & $\vYN$ & $\vU$\\
$\vN$  & $\vY$  & $\vY$ & $\vY$  & $\vY$\\
$\vYN$ & $\vYN$ & $\vN$ & $\vYN$ & $\vU$\\
$\vU$  & $\vY$  & $\vU$ & $\vYN$ & $\vU$\\
\bottomrule
\end{tabular}
\end{table}

We call $\warr$ the \emph{transition-warrant connective}: it is the
FDE-inspired $\FOUR$-valued operation we adopt to grade transitions, and it
is this operation---not any native conditional of FDE, which
$\{\lnot,\wedge,\vee\}$-based presentations of FDE need not
provide---whose per-option behaviour Proposition~\ref{thm:fdechar} certifies.
Its entries read naturally as transition warrants:
$\vY\warr\vN=\vN$ because moving from a definite commitment to its definite
negation is unwarranted (this is the retraction that
Proposition~\ref{thm:fdechar} isolates); $\vN\warr\phi=\vY$ because from a
component already false any change is unopposed; $\vY\warr\vYN=\vYN$ and
$\vY\warr\vU=\vU$ because a move from certainty into contradiction or
ignorance carries exactly that ambiguous warrant; and $\vU\warr\vU=\vU$
because ignorance about the source propagates to ignorance about the
transition.

\begin{definition}[Graded reachability]\label{def:levels}
Each level $\ell\in\{\vY,\vYN,\vU\}$ designates a set of admissible
warrants: $D_{\vY}=\{\vY\}$, $D_{\vYN}=\{\vY,\vYN\}$, and
$D_{\vU}=\{\vY,\vYN,\vU\}$. State $s'$ is reachable from $s$ by DM $i$
\emph{at level} $\ell$, written $s \to^{\ell}_i s'$, if the move is
unilateral for $i$ and $\Ent(s\to_i s') \in D_\ell$. The three levels are
read \emph{definite} (certainly executable), \emph{credible} (supported
despite possible contradiction), and \emph{possible} (not ruled out).
\end{definition}

\begin{remark}[Warrant sets, not the truth order]\label{rem:levels}
An earlier phrasing defined $s\to^{\ell}_i s'$ by
$\Ent(s\to_i s')\ge_t\ell$. That is not equivalent, and it breaks
Proposition~\ref{prop:hierarchy}: $\vU$ and $\vYN$ are
$\leqt$-incomparable, so $\Ent=\vYN$ satisfies $\Ent\ge_t\vYN$ but not
$\Ent\ge_t\vU$, and credible reachability would not be contained in
possible reachability. The Lean development has used the warrant-set
semantics from the start (\texttt{isDefiniteReach}, \texttt{isCredibleReach},
\texttt{isPossibleReach}), under which the nesting
$D_{\vY}\subset D_{\vYN}\subset D_{\vU}$ makes the hierarchy immediate;
the paper follows the machine-checked formulation.
\end{remark}

\begin{proposition}[Reachability hierarchy]\label{prop:hierarchy}
$\{s\to^{\vY}_i s'\} \subseteq \{s\to^{\vYN}_i s'\}
 \subseteq \{s\to^{\vU}_i s'\}$.
\end{proposition}

Operationally, the levels correspond to standards of strategic reasoning: a
conservative DM acts only on definite capabilities, a moderate DM on contested
but supported evidence, and an optimistic DM on anything not ruled out;
heterogeneous standards across DMs are captured by the hypergame extension of
Section~\ref{sec:hyper}.

\begin{definition}[State formula; FDE consequence]\label{def:fde}
For a state $s$ with Boolean option values, the \emph{state formula} is
$\Phi_s = \bigwedge_{i,o} \lambda_{i,o}(s)$, where
$\lambda_{i,o}(s)=p_{i,o}$ if option $o$ is selected in $s$ and
$\lambda_{i,o}(s)=\lnot p_{i,o}$ otherwise; the \emph{positive fragment}
$\Phi^{+}_s$ conjoins only the positive literals of the selected options.
An entailment is \emph{FDE-valid}, written
$\phi \vDash_{\mathrm{FDE}} \psi$, if $v(\phi) \leqt v(\psi)$ for every
four-valued valuation $v$; this order-theoretic formulation is the classical
characterization of first-degree entailment \citep{anderson1962,dunn1976}
and is equivalent to the usual presentation via preservation of the
designated values $\{\vY,\vYN\}$.
\end{definition}

\begin{remark}[A refuted characterization]\label{rem:fderefuted}
The natural conjecture---for Boolean states,
$s \to^{\vY}_i s' \iff \Phi_s \vDash_{\mathrm{FDE}} \Phi_{s'}$---is
\emph{false}, and instructively so. Table~\ref{tab:entail} makes selecting
a fresh option definite ($\vN\warr\vY=\vY$), yet
$\lnot p \nvDash_{\mathrm{FDE}} p$. Worse, state formulas are complete
conjunctions of literals, and FDE consequence between complete descriptions
holds only under literal containment, so
$\Phi_s \vDash_{\mathrm{FDE}} \Phi_{s'}$ forces $s=s'$: the conjectured
relation collapses to identity. Move warrant is not consequence between
full state descriptions.
\end{remark}

\begin{proposition}[FDE characterization of definite reachability,
corrected]\label{thm:fdechar}
For unilateral moves between states with Boolean option values, the
following are equivalent: (i) $s \to^{\vY}_i s'$; (ii) no option is
de-selected, i.e., $\mathrm{sel}(s) \subseteq \mathrm{sel}(s')$; (iii)
$\Phi^{+}_{s'} \vDash_{\mathrm{FDE}} \Phi^{+}_{s}$, where an empty
conjunction (for $\mathrm{sel}(s)=\emptyset$) is taken to be entailed by
every formula, and an empty premise ($\mathrm{sel}(s')=\emptyset$) to
entail only the empty conjunction.
\end{proposition}

\begin{proof}
(i)$\iff$(ii): for Boolean pre/post values, Table~\ref{tab:entail} gives
$\vY\warr\vY=\vN\warr\vN=\vN\warr\vY=\vY$ and $\vY\warr\vN=\vN$, so the
conjunction over changed options equals $\vY$ exactly when no option
transitions $\vY\warr\vN$. (ii)$\iff$(iii): for conjunctions of positive
literals, $\bigwedge B \vDash_{\mathrm{FDE}} \bigwedge A$ holds iff
$A\subseteq B$ as literal sets---if $A\subseteq B$, the glb over $B$ is
$\leqt$ the glb over $A$ at every valuation; if $q\in A\setminus B$, the
valuation sending $q$'s variable to $\vN$ and every variable of $B$ to
$\vY$ refutes the consequence.
\end{proof}

The corrected characterization is what licenses ``grounded in FDE'': at the
definite level a move is warranted exactly when the target \emph{contains}
the source's positive commitments---definite moves may add commitments but
never retract them, and it is retraction ($\vY\warr\vN$) that carries
warrant $\vN$. The graded levels then admit warrants from the larger sets
$D_{\vYN}$ and $D_{\vU}$ in a controlled way. Formalizing
Proposition~\ref{thm:fdechar} is the immediate next Lean target
(Section~\ref{sec:lean}); its per-option core is already fixed by the
machine-checked \texttt{Val4.entail}.

\begin{remark}[Relevance as strategic grounding]\label{rem:relevance}
FDE satisfies the variable-sharing property: $\phi\vDash_{\mathrm{FDE}}\psi$
only if $\phi$ and $\psi$ share a propositional variable
\citep{anderson1962,omori2017}. Under
Proposition~\ref{thm:fdechar} this signature feature of relevance logic
acquires a strategic meaning: full warrant requires the target configuration
to carry every positive commitment of the source---moves are contentually
anchored extensions, never wholesale replacements. Transition validity
thus respects \emph{content}, not merely truth-functional relationships; this
is what distinguishes the entailment grounding from stipulated feasibility
degrees or probabilistic weightings, which are insensitive to content.
The variable-sharing property concerns FDE formulas, so this reading is
restricted to nonempty positive fragments: when $\mathrm{sel}(s)\neq\emptyset$
and $\mathrm{sel}(s')\neq\emptyset$, condition (iii) of
Proposition~\ref{thm:fdechar} gives $\mathrm{sel}(s)\subseteq\mathrm{sel}(s')$,
so the two fragments share every variable of $\Phi^+_s$, not just one. The
empty cases fall outside the property and are settled by the convention of
Proposition~\ref{thm:fdechar}(iii) alone. If $\mathrm{sel}(s)=\emptyset$, the
source has no positive commitment, $\Phi^+_s=\bigwedge\varnothing$ is not an
FDE formula, and every move from $s$ is definitely warranted. If
$\mathrm{sel}(s)\neq\emptyset=\mathrm{sel}(s')$, the move retracts every
commitment and is not definitely warranted, in agreement with (ii).
\end{remark}

\begin{theorem}[Path entailment decay]\label{thm:decay}
For a path $\pi = s_0 \to_{i_1} s_1 \to_{i_2} \cdots \to_{i_k} s_k$ with
$\Ent(\pi) = \bigwedge_{j} \Ent(s_{j-1}\to_{i_j} s_j)$, entailment is
non-increasing under concatenation: $\Ent(\pi\cdot\pi') \le_t \Ent(\pi)$.
\end{theorem}

\begin{proof}
$\phi\wedge\psi \le_t \phi$ for all $\phi,\psi\in\FOUR$; concatenation only
adds conjuncts.
\end{proof}

Elaborate multi-step plans through epistemically uncertain territory can
carry weaker warrant than direct moves: concatenation can never increase
warrant, and a path is no more secure than its weakest constituent warrant.

\section{Quasi-Closed World Architecture}\label{sec:qcw}

The structural-safety idea is related to viability and controlled-invariance
approaches, which identify subsets of a state space from which constraints can be
maintained under admissible dynamics \citep{aubin2009}. QCW-GMCR differs in
that the transition structure is strategic and DM-dependent, while the admissible
safe region is constructed jointly from four-valued epistemic reduction and
catastrophe exclusion rather than from a control law on a continuous dynamical
system.

QCW-GMCR transforms an ambiguity-rich open world into a safety-preserving
quasi-closed world through four layers:
\[
\underbrace{\psi_{i,o}\in\FOUR}_{\text{Layer 0: options}}
\;\xrightarrow{\;\Gamma\;}\;
\underbrace{\Sopen}_{\text{Layer 1: open world}}
\;\xrightarrow{\;\delta\;}\;
\underbrace{\Sred}_{\text{Layer 2: reduced}}
\;\xrightarrow{\;\kappa\;}\;
\underbrace{\Squasi}_{\text{Layer 3: quasi-closed}}
\]

\subsection{Reduction operators}

\begin{definition}[Deterministic reduction operator]\label{def:delta}
A \emph{reduction operator} is a map $\delta:\FOUR\to\{\vY,\vN\}$ with
$\delta(\vY)=\vY$ and $\delta(\vN)=\vN$; it resolves the ambiguous values
$\vYN,\vU$ to binary decisions, formalizing the cognitive closure by which a
DM enables action under ambiguity. The reduced space is
$\Sred = \{ s \in \Sopen : \delta(\phi_s) = \vY \}$.
\end{definition}

\begin{definition}[Canonical operators]\label{def:canonical}
Table~\ref{tab:delta} lists the four canonical operators:
conservative $\dcons$, precautionary $\dprec$, optimistic $\dopt$, and
contradiction-averse $\dcont$.
\end{definition}

\begin{table}[ht]
\centering
\caption{Canonical reduction operators.}
\label{tab:delta}
\begin{tabular}{lcccl}
\toprule
Operator & $\delta(\vYN)$ & $\delta(\vU)$ & Risk attitude & Context\\
\midrule
$\dcons$ & $\vN$ & $\vN$ & risk-averse   & safety-critical\\
$\dprec$ & $\vY$ & $\vN$ & moderate      & high-stakes conflicts\\
$\dopt$  & $\vY$ & $\vY$ & risk-tolerant & exploratory\\
$\dcont$ & $\vN$ & $\vY$ & moderate      & data reconciliation\\
\bottomrule
\end{tabular}
\end{table}

\begin{definition}[Formal requirements and design principles]\label{def:admissible}
A reduction operator $\delta$ is \emph{admissible} if it meets two formal
requirements: (A0) \emph{classical extension}, $\delta(\vY)=\vY$ and
$\delta(\vN)=\vN$; and (A1) \emph{truth monotonicity}, $v_1 \leqt v_2$
and $\delta(v_1)=\vY$ imply $\delta(v_2)=\vY$. (Infeasibility preservation
$\delta(\vN)=\vN$ is part of A0, not a separate condition.) Two further
\emph{design principles} guide the choice of the two free values but are not
mathematical predicates: (P1) \emph{contradiction caution}---$\delta(\vYN)$
should be set with awareness of catastrophe proximity; and (P2)
\emph{ignorance transparency}---$\delta(\vU)$ should be stated explicitly
and subjected to sensitivity analysis. Quantification over ``admissible
operators'' elsewhere refers to A0 and A1 only.
\end{definition}

\begin{remark}[Why truth rather than knowledge monotonicity]\label{rem:a4}
An earlier formulation of the monotonicity requirement (A1) demanded
monotonicity along the knowledge order $\lek$ rather than the truth order.
The Lean formalization refuted it: $\vY \lek \vYN$ yet
$\dcons(\vY)=\vY$ and $\dcons(\vYN)=\vN$; likewise $\dopt$ fails at
$\vU \lek \vN$ and $\dcont$ at $\vY \lek \vYN$. Exhaustive case
analysis shows that $\dprec$ is the \emph{only} canonical operator that is
knowledge-monotone. Conceptually this is as it should be: moving up the
knowledge order can mean \emph{acquiring contradictory evidence}, and
contradiction caution (P1) exists precisely so that a cautious operator may
exclude a state upon such acquisition---knowledge monotonicity would
legislate that caution away. Truth monotonicity, under which inclusion is
preserved as a value becomes more definitely true, is the sound requirement,
and all four canonical operators satisfy it; both the correction and the
counterexamples are machine-checked (Section~\ref{sec:lean}).
\end{remark}

\begin{theorem}[Exhaustiveness of the canonical operators]\label{thm:uniqueness}
A0 fixes $\delta(\vY)=\vY$ and $\delta(\vN)=\vN$ and leaves only
$\delta(\vYN),\delta(\vU)\in\{\vY,\vN\}$ free, so there are exactly
$2^2=4$ classical-extension deterministic reductions---and all four satisfy
truth monotonicity (A1), hence are admissible. They are precisely
$\dcons,\dprec,\dopt,\dcont$, realizing respectively maximal caution,
contradiction tolerance with ignorance aversion, contradiction aversion with
ignorance tolerance, and maximal inclusion. The canonical operators are thus
not representative examples but an exhaustive enumeration; the normative
labels are an interpretation laid over a complete four-element space.
\end{theorem}

\begin{proposition}[Operator monotonicity]\label{prop:monotone}
$\Sred^{\dcons} \subseteq \Sred^{\dprec} \subseteq \Sred^{\dopt}$ and
$\Sred^{\dcons} \subseteq \Sred^{\dcont} \subseteq \Sred^{\dopt}$, with
$\Sred^{\dprec}$ and $\Sred^{\dcont}$ incomparable in general. This partial
order underwrites a systematic sensitivity-analysis protocol: compute
equilibria under each operator and classify them as robust (stable under all)
or fragile (operator-dependent).
\end{proposition}

\subsection{Catastrophe exclusion and the safety invariant}

\begin{definition}[Forbidden states, boundary, quasi-closed space]\label{def:kappa}
Let $F \subset \Sred$ be the \emph{forbidden} (catastrophic) states, outcomes
that must be structurally prevented rather than merely dispreferred. The
\emph{boundary} is
$B = \{ s \in \Sred\setminus F : \exists i\in N,\ \exists f \in F,\
s \to_i f \}$, and the \emph{quasi-closed state space} is
$\Squasi = \Sred \setminus (F \cup B)$.
\end{definition}

\begin{theorem}[Quasi-closed world invariant]\label{thm:invariant}
If admissible strategic paths are restricted to $\Squasi$---all DMs reason
within $\Squasi$ and take no unilateral move into $B$---then no forbidden
state is reachable from any $s \in \Squasi$.
\end{theorem}

\begin{proof}
$F\cap\Squasi = B\cap\Squasi = \emptyset$ by construction. Any path from
$\Squasi$ into $F$ must pass through the last state before entering $F$,
which by definition lies in $B$ and hence outside $\Squasi$. No such path
exists within $\Squasi$.
\end{proof}

\begin{remark}[Structural safety and strategic compliance]\label{rem:compliance}
Theorem~\ref{thm:invariant} is a conditional guarantee. It establishes
\emph{structural safety}: if moves are restricted to $\Squasi$, forbidden
states are unreachable. It does not establish \emph{strategic compliance}:
that the DMs will in fact confine themselves to $\Squasi$. Whether they do
depends on their preferences and on the stability of the states involved, and
Section~\ref{sec:example} shows that the restriction can conflict with
deterrence. The theorem should be read as a property of the restricted
analysis, not as a prediction of behaviour.
\end{remark}

\begin{remark}[Why structural exclusion]\label{rem:structural}
One may ask whether it suffices for catastrophic states to be reachable but
never \emph{stable}. Stability is a property of resting points: it constrains
where a conflict settles, not which states are traversed en route, and it
offers no protection against out-of-equilibrium behavior or misperception.
For outcomes that are unacceptable regardless of probability---irreversible by
definition---the appropriate guarantee is path-level and structural, which is
what Theorem~\ref{thm:invariant} provides; preference-based deterrence within
$\Squasi$ then operates on top of, not instead of, the structural barrier.
\end{remark}

\begin{corollary}[Defense in depth]\label{cor:depth}
Two independent barriers protect against catastrophe: the \emph{entailment
barrier} (if every transition into $F$ has $\Ent = \vN$, then $F$ is
unreachable at every level $\ell$, regardless of $\kappa$) and the
\emph{structural barrier} of boundary exclusion.
\end{corollary}

\begin{corollary}[Recovery of standard GMCR under definite
reachability]\label{cor:conservative}
Suppose $F=\emptyset$, $\phi_s\in\{\vY,\vN\}$ for all states, and every
unilateral move of the underlying GMCR between states of
$\Sred=\{s:\phi_s=\vY\}$ is definitely warranted:
$\Ent(s\to_i s')=\vY$ for all $i$ and all $(s,s')\in A_i$ with
$s,s'\in\Sred$. Then $\Squasi=\Sred$; on $\Sred$ the level sets coincide
with the underlying move sets,
$A_i^{\ell}\cap(\Sred\times\Sred)=A_i\cap(\Sred\times\Sred)$ for every level
$\ell$; and all QCW-GMCR equilibrium concepts coincide with their standard
GMCR counterparts for the graph model restricted to $\Sred$. In particular,
if $\phi_s=\vY$ for every state, they coincide with the standard concepts
for the underlying graph model itself.
\end{corollary}

\begin{remark}[The definite-warrant condition is necessary]\label{rem:recovery}
Boolean feasibility alone does not recover standard GMCR. By
Proposition~\ref{thm:fdechar}, a definite move may add positive commitments
but never retract one: an option de-selection $\vY\warr\vN$ carries warrant
$\vN$ (Table~\ref{tab:entail}) and so belongs to \emph{no} level set
$D_{\vY},D_{\vYN},D_{\vU}$. Thus a standard unilateral move that switches
an option off---$s_1=(p{=}\vY)\to_i s_0=(p{=}\vN)$ in a one-option
model---is absent from every graded reachability relation, and stability
verdicts diverge from standard GMCR unless the definite-warrant hypothesis is
imposed. This reflects a distinction between the \emph{physical} act of de-selecting
an option and its \emph{epistemic} warrant: reconciling epistemic-warrant
reachability with physical option de-selection calls for a two-layer move
model (physical option state plus epistemic warrant), which we leave to
future work.
\end{remark}

\section{Entailment-Graded and Catastrophe-Avoiding Equilibria}\label{sec:equil}

\begin{definition}[Level-$\ell$ improvements and stability]\label{def:levelstab}
For $\ell\in\{\vY,\vYN,\vU\}$, the \emph{level-$\ell$ improvement set} is
$R^{+,\ell}_i(s) = \{ s' : s\to^\ell_i s' \text{ and } s'\succ_i s\}$. State
$s$ is \emph{Nash stable at level $\ell$} if $R^{+,\ell}_i(s)=\emptyset$ for
all $i$. GMR, SMR, and SEQ stability at level $\ell$ are defined by
restricting the improvement and sanction moves of
Definition~\ref{def:stability} to level-$\ell$ reachability; an asymmetric
variant $(\ell_1,\ell_2)$ lets the initial move and the sanctions use
different levels.
\end{definition}

\begin{theorem}[Stability hierarchy]\label{thm:stabhier}
$S^{\mathrm{Nash},\vU} \subseteq S^{\mathrm{Nash},\vYN}
 \subseteq S^{\mathrm{Nash},\vY}$.
\end{theorem}

\begin{proof}
By Proposition~\ref{prop:hierarchy},
$R^{+,\vY}_i(s) \subseteq R^{+,\vYN}_i(s) \subseteq R^{+,\vU}_i(s)$; emptiness
propagates downward through the inclusions.
\end{proof}

For Nash stability, level $\vU$ is the most demanding and hence the most
robust: the state remains Nash stable even when DMs reason optimistically
about capabilities. The analogous statement fails for GMR, SMR and SEQ
(Remark~\ref{rem:levelnonmono}).

\begin{definition}[Catastrophe-avoiding equilibria]\label{def:ca}
$s'$ is \emph{safely reachable} at level $\ell$ if $s\to^\ell_i s'$ and
$s'\in\Squasi$. \emph{CA-Nash stability at level $\ell$} requires the safe
level-$\ell$ improvement set to be empty for all DMs; CA-GMR, CA-SMR, and
CA-SEQ restrict all moves---initial, sanctioning, counter-sanctioning---to
safe level-$\ell$ reachability within $\Squasi$.
\end{definition}

\begin{definition}[Robust equilibria]\label{def:robust}
Let $\sigma$ be a stability concept. A state $s$ is \emph{cross-level robust}
for $\sigma$ if it is a $\sigma$-equilibrium at each of the three levels
$\vY$, $\vYN$ and $\vU$ (for a fixed operator); \emph{cross-operator robust}
for $\sigma$ at level $\ell$ if it is a $\sigma$-equilibrium at level $\ell$
under each of the four canonical operators; and
\emph{fully robust} for $\sigma$ if it is a $\sigma$-equilibrium at every
level under every canonical operator.
\end{definition}

For $\sigma=\mathrm{Nash}$, cross-level robustness is equivalent to Nash
stability at level $\vU$, by Theorem~\ref{thm:stabhier}. For the other
concepts no such reduction holds.

\begin{remark}[Levels are not monotone for GMR, SMR and SEQ]\label{rem:levelnonmono}
Weakening the level admits more moves for every DM, and so adds sanctions as
well as improvements. Stability at level $\vU$ therefore neither implies nor
is implied by stability at level $\vY$ for GMR, SMR or SEQ. Let DM $i$ have a
definite improvement $s\to t$ and let DM $j$ have a move $t\to u$ that is
reachable only at level $\vU$, with $t\succ_i s\succ_i u$: then $s$ is GMR
stable for $i$ at level $\vU$ but not at levels $\vY$ and $\vYN$. Conversely,
if $i$'s improvement $s\to t$ is reachable only at level $\vU$ and $t$ admits
no sanction, $s$ is GMR stable for $i$ at levels $\vY$ and $\vYN$ but not at
level $\vU$. With suitable preferences the same examples apply to SMR and, if
$j$'s move is an improvement for $j$, to SEQ. Cross-level robustness for these
concepts must therefore be checked at each level, as in
Definition~\ref{def:robust}.
\end{remark}

\begin{proposition}[Robustness and safety]\label{prop:robustsafe}
Let $s$ be fully robust for a stability concept $\sigma$, and suppose every
transition into $F$ has $\Ent = \vN$. Then $s$ is a $\sigma$-equilibrium under
every admissible deterministic reduction operator and at every entailment
level, and $F$ is unreachable at every level.
\end{proposition}

\begin{proof}
The admissible deterministic reduction operators are exactly the four
canonical operators (Theorem~\ref{thm:uniqueness}), so the first claim is
Definition~\ref{def:robust}. The second is the entailment barrier of
Corollary~\ref{cor:depth}.
\end{proof}

\section{Coalitional Stability under Epistemic Ambiguity}\label{sec:coalition}

Coalition analysis extends GMCR stability to joint maneuvers by groups of DMs
\citep{kilgour2001,inohara2008a,inohara2008b}, and the interrelationships
among non-cooperative and coalitional stability definitions remain an active
topic \citep{rego2023,zhu2025}. The four-valued layer extends to this setting
with essentially no new machinery, because reachability was coalitional from
the start: the opponent reachability used in Sections~\ref{sec:prelim} and
\ref{sec:equil} is the instance $H = N\setminus\{i\}$ of reachability by an
arbitrary coalition $H$.

\begin{definition}[Level-$\ell$ coalitional improvement]\label{def:coalimp}
For a nonempty coalition $H\subseteq N$, state $s'$ is a \emph{level-$\ell$
coalitional improvement} from $s$ for $H$ if $s'$ is reachable from $s$ by a
nonempty legal sequence of unilateral moves by members of $H$ (no member moves
twice in succession; Definition~\ref{def:reach}), each carrying move
entailment in $D_\ell$, and $s' \succ_i s$ for every $i\in H$. Write
$CR^{+,\ell}_H(s)$ for the set of such states. Requiring every member of
$H$ to improve follows the coalitional improvements of \citet{kilgour2001}
and \citet{inohara2008a}: a coalition is formed only if joining benefits each
member. Requiring improvement only for the members that actually move would
give a larger set of coalitional improvements, in which passive members may
be made worse off; this choice is orthogonal to the four-valued layer, which
grades the constituent moves.
\end{definition}

\begin{definition}[Level-$\ell$ coalitional Nash stability]\label{def:cnash}
$s$ is \emph{coalitionally Nash stable at level $\ell$} if
$CR^{+,\ell}_H(s)=\emptyset$ for every nonempty $H\subseteq N$. The
\emph{catastrophe-avoiding} variant restricts every constituent move to safe
level-$\ell$ reachability within $\Squasi$.
\end{definition}

\begin{proposition}[Coalitional hierarchy]\label{prop:coalhier}
$CR^{+,\vY}_H(s) \subseteq CR^{+,\vYN}_H(s) \subseteq CR^{+,\vU}_H(s)$
for every $H$; hence the coalitionally Nash-stable sets satisfy
$S^{\mathrm{CNash},\vU} \subseteq S^{\mathrm{CNash},\vYN} \subseteq
S^{\mathrm{CNash},\vY}$.
\end{proposition}

\begin{proof}
The warrant condition on each constituent move relaxes along the nested
warrant sets $D_{\vY}\subset D_{\vYN}\subset D_{\vU}$
(Proposition~\ref{prop:hierarchy}); emptiness then propagates as in
Theorem~\ref{thm:stabhier}.
\end{proof}

\begin{remark}[Epistemic fragility of coalitional maneuvers]\label{rem:coalfragile}
Coalitional improvements are multi-step by nature, so path entailment decay
(Theorem~\ref{thm:decay}) applies with particular force: the warrant of a
joint maneuver is bounded by its weakest link. A coalition that appears
powerful under binary analysis may be unable to execute any \emph{warranted}
joint plan at the credible or definite level: tightening the level
(shrinking $D_\ell$) can only remove coalitional improvements
(Proposition~\ref{prop:coalhier}), and a joint maneuver is lost as soon as
any one of its constituent moves falls outside $D_\ell$.
Four-valued analysis thus separates a coalition's nominal power from its
epistemically executable power---a distinction invisible both to binary
coalition analysis and to individual-level four-valued analysis.
\end{remark}

\begin{proposition}[The invariant is coalition-robust]\label{prop:coalinv}
Theorem~\ref{thm:invariant} extends verbatim to coalitional maneuvers: if
all DMs restrict reasoning to $\Squasi$ and avoid unilateral moves into $B$,
then no sequence of unilateral moves by any coalition reaches $F$.
\end{proposition}

\begin{proof}
The proof of Theorem~\ref{thm:invariant} is path-structural: any sequence of
unilateral moves ending in $F$ passes through the last state before entering
$F$, which lies in $B$ by definition. No appeal is made to who moves, so
coalitional sequences are covered without a separate safety argument.
\end{proof}

The classical (binary) core of this section---coalitional Nash, GMR, SMR, and
SEQ stability with class-based opponent responses, and the hierarchy
$\mathrm{CNash}\subseteq\mathrm{CSMR}\subseteq\mathrm{CGMR}$ and
$\mathrm{CNash}\subseteq\mathrm{CSEQ}\subseteq\mathrm{CGMR}$---is
machine-checked for unrestricted coalition reachability
(Section~\ref{sec:lean}); the paper-level proofs of these inclusions do not use
the legality constraint and therefore also cover the legal sequences of
Definition~\ref{def:coalimp}. The legal versions and the level-$\ell$ grading
defined above are the formalization frontier.

\section{Four-Valued Hypergames}\label{sec:hyper}

DMs may disagree about feasibility itself. In a \emph{four-valued hypergame},
each DM $i$ holds subjective assessments $\psi^i_{j,o}\in\FOUR$ of every
option of every DM, inducing a subjective feasibility $\phi^i_s$, a subjective
reduced space via $i$'s operator $\delta_i$, and a subjective quasi-closed
world $S^i_{\mathrm{qc}}$. The \emph{shared safe core} is
$S^{\mathrm{core}} = \bigcap_{i\in N} S^i_{\mathrm{qc}}$.

\begin{definition}[Hypergame CA-equilibrium]\label{def:hyperca}
Let $\sigma$ be a stability concept. For each DM $i$, let $G^i$ be the
subjective game of $i$ restricted to $S^{\mathrm{core}}$: its states are those
of $S^{\mathrm{core}}$, and its moves and preferences are those of $i$'s
subjective game between states of $S^{\mathrm{core}}$. A state
$s\in S^{\mathrm{core}}$ is a \emph{hypergame CA-$\sigma$ equilibrium} if, for
every DM $i$, $s$ is $\sigma$-stable for $i$ in $G^i$.
\end{definition}

The restriction to $S^{\mathrm{core}}$ is part of the definition: moves,
sanctions and counter-sanctions that leave the shared safe core are not
considered, and full-game stability is neither claimed nor required.

\begin{remark}[$S^{\mathrm{core}}\neq\emptyset$ does not suffice for
existence]\label{rem:hypercex}
The source framework claimed existence of a hypergame CA-equilibrium from
$S^{\mathrm{core}}\neq\emptyset$ alone, for any stability concept
admitting a stable state on every finite graph. The claim fails: per-game
existence does not make the per-DM stable sets intersect across
\emph{different} subjective games. Two states $\{a,b\}$ and two DMs
suffice: in DM 1's subjective game only DM 1 can move
($a\leftrightarrow b$) and $b\succ_1 a$; in DM 2's subjective game only
DM 2 can move and $a\succ_2 b$; let both states be in every subjective
quasi-closed world, so that $S^{\mathrm{core}}=\{a,b\}$ and
$G^i$ is DM $i$'s whole subjective game. Each subjective game has an equilibrium,
but the only state stable for DM 1 in $G^1$ is $b$, while the only
state stable for DM 2 in $G^2$ is $a$---for Nash and, sanctions being
absent, equally for GMR, SMR, and SEQ---so no hypergame CA-equilibrium
exists, although $S^{\mathrm{core}}\neq\emptyset$. The counterexample is finite
and fully explicit, hence an immediate candidate for machine-checking.
\end{remark}

\begin{theorem}[Existence under agreement on the core]\label{thm:hyperexist}
Suppose $S^{\mathrm{core}}\neq\emptyset$ and all DMs' subjective structures
(reachability and preferences) coincide on $S^{\mathrm{core}}$. If the
resulting CA-restricted graph model $\tilde G$ admits an equilibrium under
stability concept $\sigma$, then a hypergame CA-$\sigma$ equilibrium exists
in $S^{\mathrm{core}}$.
\end{theorem}

\begin{proof}
Under agreement, the subjective structures restricted to $S^{\mathrm{core}}$
define a single graph model $\tilde G$---the CA-restricted game, in which
only safe moves within $S^{\mathrm{core}}$ are available. A $\sigma$-equilibrium
$s^*$ of $\tilde G$ is $\sigma$-stable for every DM in $\tilde G=G^i$,
which is exactly the condition of Definition~\ref{def:hyperca}; and
$s^* \in S^{\mathrm{core}}$ is recognized as feasible and safe by all DMs.
So $s^*$ is a hypergame CA-$\sigma$ equilibrium.
\end{proof}

Characterizing the weakest epistemic-divergence conditions under which
existence survives---strictly between full agreement on the core, which
suffices, and bare nonemptiness of the core, which does not---is an open
problem that Remark~\ref{rem:hypercex} makes concrete.

\begin{remark}[Relation to hypergame-GMCR]\label{rem:hyperrel}
Hypergame analysis within GMCR is a developed line: first-level hypergames in
graph form model misperception held by and about DMs \citep{aljefri2018};
general hypergame analysis handles multiple perception levels and
self-misperception through universal sets of options and states
\citep{aljefri2020}; and recent work combines first-level hypergames with
\emph{unknown} preference misperceptions, using algebraic representations for
efficient stability computation \citep{zhao2024}; see \citet{trencsenyi2025}
for a broad survey. In all of these, every perceived game remains binary:
uncertainty lives in \emph{which} game a DM perceives, or in preferences,
while each option within each perceived game is committed to $\vY$ or $\vN$.
The four-valued hypergame moves the epistemic resolution inside the
assessment itself: $\psi^i_{j,o}=\vYN$ records that DM $i$ holds
contradictory evidence about an option---inexpressible in a binary perceived
game---and $\vU$ separates ignorance from rejection without multiplying
perceived games. The model degenerates gracefully: binary and coincident
assessments yield an ordinary graph model, and binary but divergent
assessments yield a first-level hypergame in the classical sense, so the
construction strictly generalizes hypergame-GMCR along the epistemic axis.
Where the universal-set construction of \citet{aljefri2018} and
\citet{aljefri2020} already distinguishes states recognized by all DMs from
individually perceived ones, the shared safe core sharpens that object to the
intersection of subjective \emph{quasi-closed} worlds---states each DM
recognizes as feasible \emph{and} safe under its own reduction
operator---so that $S^{\mathrm{core}}=\emptyset$ acquires a diagnostic
reading (structural epistemic conflict requiring belief alignment before
equilibrium analysis) that is unavailable when perception is the only
epistemic variable.

The construction also relates to \emph{interactive unawareness}
\citep{rego2020}, where a DM may be unaware of options available to
itself or others. That line treats awareness structurally by determining
which options enter a DM's subjective model, including iterated levels
of awareness about others' awareness. QCW-GMCR addresses a distinct
but complementary question for options already represented in that
model: $\psi^i_{j,o}=\vU$ records that DM $i$ has grounds for neither
accepting nor rejecting the status of option $o$ of DM $j$.
Thus structural unawareness concerns \emph{which} options are present
in a subjective game, whereas $\vU$ concerns the epistemic standing of
an option already in view; moreover, $\vYN$ represents contradictory
evidence, which has no direct counterpart in an awareness-only account.
\end{remark}

When $S^{\mathrm{core}}=\emptyset$, no state is simultaneously recognized as
feasible and safe by all DMs; this \emph{diagnoses} a structural epistemic
conflict that must be resolved by belief alignment (joint fact-finding,
operator harmonization, catastrophe-threshold alignment) before equilibrium
analysis can proceed.

\section{Illustrative Example}\label{sec:example}

\paragraph{Model.}
Two nations $A$ and $B$ each choose a posture, \emph{De-escalate} ($D$) or
\emph{Escalate} ($E$). Intelligence yields definite assessments of $A$'s two
postures and of $B$'s de-escalation, $\psi_{A,\mathrm{de}}=\psi_{A,\mathrm{esc}}=\psi_{B,\mathrm{de}}=\vY$,
but contradictory evidence about $B$'s escalation capability,
$\psi_{B,\mathrm{esc}}=\vYN$. Each nation is represented by a single
posture variable, whose value in a state is the assessment of the posture
chosen there; the state-level feasibility $\phi_s$ is the conjunction of the
two posture values (Definition~\ref{def:comp}), and a move changes one posture
variable, so that its warrant is the pre-value $\warr$ post-value of that
variable (Definition~\ref{def:entail}). (If each posture were instead a
separate Boolean option, selecting $E$ would carry warrant
$\vN\warr\vY=\vY$, and the ambiguity about $B$'s capability would not grade
the move; the posture reading is the one under which it does.)

The states are $DD$, $ED$, $DE$, $EE$ (with $A$'s posture first) and the
forbidden state $F$ (strategic nuclear exchange). Each nation can switch its
own posture in either direction, and either nation can move from $EE$ to $F$,
which is absorbing; the moves into $F$ are taken to be fully warranted. The
feasibility values and move warrants are given in Table~\ref{tab:exmodel} and
the graph in Figure~\ref{fig:exgraph}. The preferences are
\[
A:\ DD \succ EE \succ ED \succ DE \succ F,\qquad
B:\ DE \succ DD \succ ED \succ EE \succ F.
\]
$A$ prefers the status quo and answers an escalation by $B$ with its own
rather than accept it; $B$ would gain from escalating unilaterally but
prefers mutual de-escalation to mutual escalation.

\begin{table}[ht]
\centering
\caption{The illustrative model: feasibility of the states and warrants of the moves.}
\label{tab:exmodel}
\begin{tabular}{@{}lc|lll@{}}
\toprule
State & $\phi_s$ & Move & Mover & Warrant\\
\midrule
$DD$ & $\vY$   & $DD\leftrightarrow ED$,\ $DE\leftrightarrow EE$ & $A$ & $\vY\warr\vY=\vY$\\
$ED$ & $\vY$   & $DD\to DE$,\ $ED\to EE$ & $B$ & $\vY\warr\vYN=\vYN$\\
$DE$ & $\vYN$  & $DE\to DD$,\ $EE\to ED$ & $B$ & $\vYN\warr\vY=\vYN$\\
$EE$ & $\vYN$  & $EE\to F$ & $A$ or $B$ & $\vY$ (assumed)\\
$F$  & $\vY$   & & & \\
\bottomrule
\end{tabular}
\end{table}

\begin{figure}[ht]
\centering
\begin{tikzpicture}[x=2.6cm,y=1.7cm,>={Stealth[length=2.2mm]},
  st/.style={circle,draw,minimum size=9mm,inner sep=0pt}]
  \node[st] (DD) at (0,1) {$DD$};
  \node[st] (ED) at (1,1) {$ED$};
  \node[st] (DE) at (0,0) {$DE$};
  \node[st] (EE) at (1,0) {$EE$};
  \node[st,fill=gray!20] (F) at (2.1,0) {$F$};
  \draw[<->,thick] (DD) -- node[above,font=\small]{$A$} (ED);
  \draw[<->,thick] (DE) -- node[below,font=\small]{$A$} (EE);
  \draw[<->,thick,dashed] (DD) -- node[left,font=\small]{$B$} (DE);
  \draw[<->,thick,dashed] (ED) -- node[right,font=\small]{$B$} (EE);
  \draw[->,thick,dotted] (EE) -- node[above,font=\small]{$A$, $B$} (F);
\end{tikzpicture}
\caption{Graph of the illustrative model. Solid: moves of $A$ (warrant $\vY$);
dashed: moves of $B$ (warrant $\vYN$); dotted: irreversible moves into the
forbidden state $F$. All posture moves are reversible.}
\label{fig:exgraph}
\end{figure}

\paragraph{Reduction and the quasi-closed space.}
Under $\dcons$ and $\dcont$, $\delta(\vYN)=\vN$, so $DE$ and $EE$ leave the
reduced space: $\Sred=\{DD,ED,F\}$, no state of $\Sred$ leads to $F$, and
$\Squasi=\{DD,ED\}$. Under $\dprec$ and $\dopt$, all five states are in
$\Sred$, the boundary of Definition~\ref{def:kappa} is $\{EE\}$ (written
without the symbol $B$ to avoid confusion with nation $B$), and
$\Squasi=\{DD,ED,DE\}$.

\paragraph{Stability of mutual de-escalation.}
The verdicts for $DD$ are summarized in Table~\ref{tab:exresults}; they were
checked by exhaustive computation, with legal and with unrestricted
reachability, which give the same verdicts in this model.
\begin{itemize}
\item \emph{Definite level.} $B$'s escalation has warrant $\vYN\notin D_{\vY}$,
so $B$ has no escalation move, and $A$ does not prefer $ED$ to $DD$. Hence
$DD$ is Nash stable, and so stable under every concept, for every operator.
\item \emph{Credible and possible levels, $\dcons$ and $\dcont$.} $DE$ is not
in $\Sred$, so $B$ again has no improvement from $DD$.
\item \emph{Credible and possible levels, $\dprec$ and $\dopt$.} The threat
re-enters: $DD\to DE$ is an improvement for $B$, so $DD$ is not Nash stable. It
is sanctioned by $A$'s counter-escalation $DE\to EE$, since $DD\succ_B EE$;
the sanction is an improvement for $A$ ($EE\succ_A DE$), so it also counts for
SEQ; and from $EE$, $B$ can only move to $ED$ or $F$, neither of which it
prefers to $DD$, so it counts for SMR. Hence $DD$ is GMR, SMR and SEQ stable.
\end{itemize}
Mutual de-escalation is therefore a GMR, SMR and SEQ equilibrium under all
four operators and at all three levels: a \emph{fully robust} equilibrium in
the sense of Definition~\ref{def:robust} for these concepts, though not for
Nash stability.

\begin{table}[ht]
\centering
\caption{Stability of $DD$ for both nations (Nash / GMR / SMR / SEQ;
$\checkmark$ stable, $\times$ not stable), without and with catastrophe
avoidance. Levels: definite ($\vY$), credible ($\vYN$), possible ($\vU$).}
\label{tab:exresults}
\begin{tabular}{@{}l|ccc|ccc@{}}
\toprule
& \multicolumn{3}{c|}{Standard concepts} & \multicolumn{3}{c}{Catastrophe-avoiding concepts}\\
Operator & $\vY$ & $\vYN$ & $\vU$ & $\vY$ & $\vYN$ & $\vU$\\
\midrule
$\dcons$, $\dcont$ & \checkmark\checkmark\checkmark\checkmark & \checkmark\checkmark\checkmark\checkmark & \checkmark\checkmark\checkmark\checkmark
 & \checkmark\checkmark\checkmark\checkmark & \checkmark\checkmark\checkmark\checkmark & \checkmark\checkmark\checkmark\checkmark\\
$\dprec$, $\dopt$ & \checkmark\checkmark\checkmark\checkmark & $\times$\checkmark\checkmark\checkmark & $\times$\checkmark\checkmark\checkmark
 & \checkmark\checkmark\checkmark\checkmark & $\times\times\times\times$ & $\times\times\times\times$\\
\bottomrule
\end{tabular}
\end{table}

\paragraph{Catastrophe avoidance and deterrence.}
Under the catastrophe-avoiding concepts (Definition~\ref{def:ca}) the picture
changes. $EE$ lies on the boundary, one step from $F$, so it is not in
$\Squasi$ and $A$'s counter-escalation is not safely reachable. Under $\dprec$
and $\dopt$ at the credible and possible levels, $B$'s escalation is then an
unsanctioned improvement, and $DD$ is not CA-GMR, CA-SMR or CA-SEQ stable.
The structural barrier of Theorem~\ref{thm:invariant} excludes exactly the
state through which deterrence operates: when the deterrent threat passes
through the boundary, catastrophe avoidance and deterrence pull in opposite
directions, and the catastrophe-avoiding equilibrium concepts do not refine
the standard ones but can remove an equilibrium that rests on such a threat.
Remark~\ref{rem:structural} places preference-based deterrence on top of the
structural barrier; the example shows that this is possible only for threats
that stay within $\Squasi$. (Proposition~\ref{prop:robustsafe} does not
apply here, since the moves into $F$ are warranted.)

\paragraph{Legal and unrestricted reachability.}
The verdicts above do not depend on the legality constraint, but the example
shows where it can matter. If $B$ preferred $EE$ to $DD$, $A$'s
counter-escalation alone would not be a sanction; with unrestricted
reachability $A$ could follow it with a move into $F$, two moves in
succession, and $F$ would count as a sanction, whereas with legal
reachability (Definition~\ref{def:reach}) it would not.
The four-valued state definition underlying this example goes back to
\citet{kato2022}.

\section{Machine-Checked Formalization in Lean 4}\label{sec:lean}

The accompanying repository formalizes the framework's core in Lean~4 with
mathlib, in the file \texttt{Basic.lean}. The development has two parts: a formalization
of classical GMCR that is reusable independently of the four-valued extension,
and the four-valued layer specific to QCW-GMCR. A companion file in the same
repository, \texttt{Robust.lean}, adds the stability concepts with the legal
reachability of Definition~\ref{def:reach} (Section~\ref{sec:lean:legal}); it
uses Lean's core library only. All results below compile without
\texttt{sorry}: every listed theorem has a machine-checked proof in Lean~4,
relative to the stated library dependencies (mathlib for \texttt{Basic.lean},
the core library alone for \texttt{Robust.lean}). The
formalization has already earned its keep as more than certification: it
falsified one axiom of the framework as originally stated and forced its
correction (Remark~\ref{rem:a4}).

\subsection{Classical GMCR core}

A graph model is a structure over arbitrary types of DMs and states, with
moves and strict preference as raw relations:

\begin{lstlisting}
structure GraphModel (N : Type u) (S : Type v) where
  move : N -> S -> S -> Prop
  pref : N -> S -> S -> Prop
\end{lstlisting}

Unilateral improvement, coalitional reachability (as a transitive closure of
single moves by coalition members, for an \emph{arbitrary} coalition
$H\subseteq N$---the opponent set of a single DM is the instance
$H=N\setminus\{i\}$, so the coalitional layer of
Section~\ref{sec:coalition} is definable with no new infrastructure), the
four stability concepts, and equilibrium are then defined, and the classical
hierarchy is proved. In \texttt{Basic.lean} this reachability is
\emph{unrestricted}: the same DM may move twice in succession, unlike in
Definition~\ref{def:reach}. The results are:
$\mathrm{Nash}\subseteq\mathrm{SMR}\subseteq\mathrm{GMR}$,
$\mathrm{Nash}\subseteq\mathrm{SEQ}\subseteq\mathrm{GMR}$, improvement
reachability entails reachability ($R^+_H \subseteq R_H$), and Nash
equilibria are GMR equilibria. A variant of coalitional reachability with the
Fang--Hipel--Kilgour legality constraint (no DM moves twice in succession) is
formalized by indexing move sequences with the last mover, and legal
reachability is proved to entail unrestricted reachability.

\paragraph{Legal reachability.}\label{sec:lean:legal}
\texttt{Robust.lean} states Nash, GMR, SMR and SEQ stability with legal
sequences of the opponents' moves, as in Definition~\ref{def:stability}, and
proves: the hierarchy
$\mathrm{Nash}\subseteq\mathrm{SMR}\subseteq\mathrm{GMR}$ and
$\mathrm{Nash}\subseteq\mathrm{SEQ}\subseteq\mathrm{GMR}$ for legal and for
unrestricted reachability; that legal reachability entails unrestricted
reachability and, consequently, that stability with legal reachability implies
stability with unrestricted reachability for all four concepts; and, in a
four-state model in which an opponent can move $t\to u\to v$ but not
$t\to v$, that the converse fails: a state is GMR stable with unrestricted but
not with legal reachability. These results are stated for arbitrary move
relations, and hence apply to the proper moves of Definition~\ref{def:gmcr}.
The hierarchy proofs use no axioms. The coalitional concepts are formalized
in \texttt{Basic.lean} with unrestricted reachability only.

Building directly on the same coalition-general reachability, the classical
coalition stability concepts of Section~\ref{sec:coalition} are formalized:
coalitional Nash, GMR, SMR, and SEQ stability (for coalitions, for DMs, and
as equilibria), with opponents responding through the class of nonempty
coalitions disjoint from $H$ and with an explicit indifference relation
alongside strict preference. The machine-checked results comprise the
coalitional hierarchy
$\mathrm{CNash}\subseteq\mathrm{CSMR}\subseteq\mathrm{CGMR}$ and
$\mathrm{CNash}\subseteq\mathrm{CSEQ}\subseteq\mathrm{CGMR}$ together
with the class-level inclusion
$R_{\hat H}^{++} \subseteq R_{\hat H}$ (class improvement sequences are
class move sequences). Notably, the development declares the weak-order
axioms for preference and indifference, yet none of the inclusion proofs
uses them: the coalitional hierarchy, like the individual one, is verified
for arbitrary preference and indifference relations.

Two design points deserve emphasis. First, \emph{no axioms are imposed on
preferences}: the hierarchy theorems are verified for arbitrary strict
preference relations, so they require neither completeness nor transitivity of
$\succ_i$---the standard complete, possibly intransitive preferences that
GMCR permits are a special case. Second, all definitions are stated over arbitrary
(possibly infinite) type-level models; finiteness is imposed only where
computation requires it.

\subsection{Four-valued layer}

The truth space is an inductive type with decidable equality, and the
connectives are total functions, so every algebraic fact is decidable and
proved by exhaustive case analysis certified by the kernel:

\Needspace{8\baselineskip}
\begin{lstlisting}
inductive Val4 | N | U | YN | Y
  deriving DecidableEq, Repr

def Val4.conj : Val4 -> Val4 -> Val4   -- Table 1
def Val4.entail : Val4 -> Val4 -> Val4 -- Table 2
def Val4.le_k : Val4 -> Val4 -> Prop   -- knowledge order
\end{lstlisting}

Verified results include: commutativity, associativity, idempotence of
$\wedge$, absorption by $\vN$, identity of $\vY$ (Remark~\ref{rem:algebra});
compositional propagation over arbitrary finite option sets, cases (i) and
(iv) of Proposition~\ref{prop:propagation}, with $\Gamma_\wedge$ realized as
a fold over a \texttt{Fintype} of options; the four canonical reduction
operators as computable functions
$\dcons,\dprec,\dopt,\dcont : \texttt{Val4} \to \texttt{Bool}$ together with
classical extension (A0) and truth monotonicity (A1) for each,
together with the refutation of knowledge monotonicity for $\dcons$,
$\dopt$, and $\dcont$ and its proof for $\dprec$
(Definition~\ref{def:admissible}, Remark~\ref{rem:a4}); the graded reachability predicates
\texttt{isDefiniteReach}, \texttt{isCredibleReach}, \texttt{isPossibleReach}
at the transition-value level with the hierarchy of
Proposition~\ref{prop:hierarchy} (definite $\Rightarrow$ credible
$\Rightarrow$ possible); the quasi-closed space as a computable
\texttt{Finset} difference $\Sred \setminus (F\cup B)$
(Definition~\ref{def:kappa}); and the subjective assessment/preference core
of the hypergame structure, together with a shared-safe-core membership
predicate over supplied subjective state spaces (Section~\ref{sec:hyper}).

Table~\ref{tab:corr} summarizes the correspondence between the paper and the
Lean development.

\begin{table}[htbp]
\centering
\caption{Paper--Lean correspondence. All listed items are machine-checked.}
\label{tab:corr}

\scriptsize
\renewcommand{\arraystretch}{1.12}

\begin{tabularx}{\textwidth}{
@{}
>{\raggedright\arraybackslash}p{0.36\textwidth}
>{\raggedright\arraybackslash}X
@{}
}
\toprule
\textbf{Paper statement} & \textbf{Lean identifier(s)} \\
\midrule

Definitions~\ref{def:gmcr} and~\ref{def:stability}, with unrestricted reachability (\texttt{Basic.lean})
&
\path{GraphModel},
\path{UI},
\path{CMove},
\path{NashStable}, \ldots
\\

Definitions~\ref{def:reach} and~\ref{def:stability}, with legal reachability (\texttt{Robust.lean})
&
\path{LReach},
\path{Stable}
\\

Nash $\subseteq$ SMR $\subseteq$ GMR
&
\path{NashStable.smrStable},
\path{SMRStable.gmrStable}
\\

Nash $\subseteq$ SEQ $\subseteq$ GMR
&
\path{NashStable.seqStable},
\path{SEQStable.gmrStable}
\\

$R_H^{+} \subseteq R_H$
&
\path{cui_imp_cmove}
\\

Legal $\subseteq$ unrestricted reachability
&
\path{lmove_imp_cmove};
\path{LReach.toUReach}
\\

Nash $\subseteq$ SMR $\subseteq$ GMR and Nash $\subseteq$ SEQ $\subseteq$ GMR, legal reachability (\texttt{Robust.lean})
&
\path{nash_imp_smr},
\path{smr_imp_gmr},
\path{nash_imp_seq},
\path{seq_imp_gmr}
\\

Legal stability $\Rightarrow$ unrestricted stability; the converse fails (\texttt{Robust.lean})
&
\path{stable_legal_imp_unrestricted};
\path{LegalCounterexample.unrestricted_gmr},
\path{LegalCounterexample.not_legal_gmr}
\\

Coalitional stability
(Section~\ref{sec:coalition}, classical core; unrestricted reachability)
&
\path{CNashStableForCoalition},
\path{CGMRStableForCoalition},
\path{CSMRStableForCoalition},
\path{CSEQStableForCoalition}
\\

CNash $\subseteq$ CSMR $\subseteq$ CGMR;
CNash $\subseteq$ CSEQ $\subseteq$ CGMR
&
\path{cnash_implies_csmr},
\path{csmr_implies_cgmr},
\path{cseq_implies_cgmr}
\\

$R_{\hat H}^{++} \subseteq R_{\hat H}$
&
\path{R_C_dpp_subset_R_C}
\\

Table~\ref{tab:conj}, Remark~\ref{rem:algebra}
&
\path{Val4.conj},
\path{conj_comm},
\path{conj_assoc},
\path{conj_idem},
\path{conj_N_left},
\path{conj_Y_left}
\\

Proposition~\ref{prop:propagation}(i), (iv)
&
\path{propagation_property_i},
\path{propagation_property_iv}
\\

Table~\ref{tab:delta}; classical-extension / reduction properties
&
\path{delta_cons},
\path{delta_prec},
\path{delta_opt},
\path{delta_cont}
\\

Definition~\ref{def:admissible} (A1): truth monotonicity
&
\path{cons_truth_monotone},
\path{prec_truth_monotone},
\path{opt_truth_monotone},
\path{cont_truth_monotone}
\\

Remark~\ref{rem:a4}
(knowledge-monotonicity results)
&
\path{cons_not_knowledge_monotone},
\path{opt_not_knowledge_monotone},
\path{cont_not_knowledge_monotone},
\path{prec_satisfies_knowledge_monotonicity}
\\

Table~\ref{tab:entail}
(transition-warrant connective)
&
\path{Val4.entail}
\\

Proposition~\ref{prop:hierarchy}
(value-level reachability hierarchy)
&
\path{definite_imp_credible},
\path{credible_imp_possible}
\\

Definition~\ref{def:kappa} ($S_{\mathrm{quasi}}=S_{\mathrm{reduced}}\setminus(F\cup B)$ part)
&
\path{qcwStateSpace}
\\

Hypergame assessment core; shared-core membership
&
\path{QcwHypergame},
\path{is_in_shared_safe_core}
\\

\bottomrule
\end{tabularx}
\end{table}

\subsection{Scope and outlook}

The formalization deliberately targets the load-bearing skeleton: the
algebra that drives propagation, the formal requirements governing the
canonical reduction operators, the reachability hierarchy, and the classical
stability hierarchy that Corollary~\ref{cor:conservative} bridges to. Not yet
formalized are the corrected FDE characterization
(Proposition~\ref{thm:fdechar}) and the hypergame existence
counterexample (Remark~\ref{rem:hypercex})---both proved on paper here and
next in line for the Lean development---the quasi-closed world invariant
(Theorem~\ref{thm:invariant}), path entailment decay
(Theorem~\ref{thm:decay}), the level-$\ell$ stability hierarchy over full
graph models (Theorem~\ref{thm:stabhier}), the level-$\ell$ grading of the
coalitional concepts of Section~\ref{sec:coalition} with their
interrelationship theorems in the style of \citet{inohara2008b} and
\citet{zhu2025}, and the corrected hypergame existence theorem
(Theorem~\ref{thm:hyperexist}); these are the immediate next targets for the
Lean development. Because the four-valued layer is built from decidable,
computable definitions over \texttt{Fintype}s, the same development doubles as
a verified reference implementation: concrete conflicts can be analyzed with
\texttt{\#eval}/\texttt{decide}, so that the sensitivity-analysis protocol of
Proposition~\ref{prop:monotone} runs on code whose correctness properties are
theorems rather than test cases. This matches the framework's motivating
concern: in safety-critical decision support, the analysis machinery itself
should carry formal guarantees.

\section{Concluding Remarks}\label{sec:conclusion}

This note recorded the core of QCW-GMCR---option-level four-valued
foundations, entailment-graded reachability, formally constrained reduction operators,
and catastrophe-avoiding equilibria with the quasi-closed world invariant---
together with the current status of its Lean~4 formalization. The
formalization exposed the reduction-operator monotonicity issue
(Remark~\ref{rem:a4}) and helped align the reachability definitions with the
machine-checked warrant hierarchy (Remark~\ref{rem:levels}). Further
analytical checks led to the corrected FDE characterization, the GMCR-recovery
condition, and the hypergame-existence statement
(Remarks~\ref{rem:fderefuted}, \ref{rem:recovery}, and~\ref{rem:hypercex}). The
illustrative example of Section~\ref{sec:example} also shows that structural
safety and deterrence can conflict: when a deterrent threat passes through the
boundary, restricting play to $\Squasi$ removes the sanction that sustains the
status quo (Remark~\ref{rem:compliance}). Future work
includes formalizing the remaining theorems listed above, dynamic extensions
with belief revision, empirical validation of reduction operators, and
scaling methods for large option spaces.

\subsection*{Acknowledgments}
This research did not receive any specific grant from funding agencies in the
public, commercial, or not-for-profit sectors.

\end{document}